\documentclass{article}
 
\usepackage{amsmath, amsthm, amssymb, hyperref}
\usepackage{adjustbox}
\usepackage{tikz}
\usepackage[caption=false]{subfig}

\title{XPL: An extended probabilistic logic \\ for probabilistic
  transition systems\footnote{This work was partially supported by NSF
    grant IIS-1447549.}}

\author{ANDREY GORLIN \\ C. R. RAMAKRISHNAN \\ Department of Computer
  Science, \\ Stony Brook University, Stony Brook, NY 11794,
  U.S.A. \\ $\{$\texttt{agorlin,cram$\}$@cs.stonybrook.edu}}

\theoremstyle{plain}
\newtheorem{theorem}{Theorem}
\newtheorem{lemma}[theorem]{Lemma}
\newtheorem{corollary}[theorem]{Corollary}
\theoremstyle{definition}
\newtheorem{definition}[theorem]{Definition}
\newtheorem{example}[theorem]{Example}
\theoremstyle{remark}
\newtheorem*{remark}{Remark}

\newcommand{\diam}[1]{\langle #1 \rangle}
\newcommand{\fcn}[2][\rightarrow] {\mathrel{:} #2 #1}

\newcommand{\tarrow}[1] {\stackrel{#1} \rightarrow}

\newcommand{\id}[1]{\mbox{\it #1\/}}
\newcommand{\kw}[1]{\mbox{\sf #1}}
\newcommand{\fpe}{\id{FPE}}
\newcommand{\probabs}{\id{PPA}}
\newcommand{\groupmodalities}{\id{GRP}}
\newcommand{\depgr}{\kw{Dg}}
\newcommand{\bbbn}{\ensuremath{\mathbb{N}}}

\usetikzlibrary{automata,positioning}

\begin{document}

\maketitle

\begin{abstract}
Generalized Probabilistic Logic (GPL) is a temporal logic, based on
the modal mu-calculus, for specifying properties of reactive
probabilistic systems.  We explore XPL, an extension to GPL allowing
the semantics of nondeterminism present in Markov decision processes
(MDPs).  XPL is expressive enough that a number of independently
studied problems--- such as termination of Recursive MDPs (RMDPs),
PCTL* model checking of MDPs, and reachability for Branching MDPs---
can all be cast as model checking over XPL.  Termination of multi-exit
RMDPs is undecidable; thus, model checking in XPL is undecidable in
general.  We define a subclass, called separable XPL, for which model
checking is decidable.  Decidable problems such as termination of
1-exit RMDPs, PCTL* model checking of MDPs, and reachability for
Branching MDPs can be reduced to model checking separable XPL.  Thus,
XPL forms a uniform framework for studying problems involving systems
with non-deterministic and probabilistic behaviors, while separable
XPL provides a way to solve decidable fragments of these problems.
\end{abstract}

\section{Introduction}
% Importance of the mu-calculus and its extensions to probabilistic
% systems
For finite-state systems, model checking a temporal property can be
cast in terms of model checking in the modal $\mu$-calculus, the
so-called ``assembly language'' of temporal logics.   A number of
temporal logics have been proposed and used for specifying properties
of finite-state \emph{probabilistic} systems.  Two of the notable
logics for probabilistic systems based on the $\mu$-calculus are
GPL~\cite{gpl} and pL$\mu$~\cite{indp}.  

% GPL, its features and compromises needed to get them.
GPL is defined over Reactive Probabilistic Labeled Transition Systems
(RPLTSs).  In an RPLTS, each state has a set of outgoing transitions
with distinct labels; each transition, in turn, specifies a
(probabilistic) distribution of target states.  The branching-time
probabilistic logic GPL is expressive enough to serve as an ``assembly
language'' of a large number of probabilistic temporal logics.  For
instance, model checking PCTL* properties over Markov Chains, as well
as termination and reachability of \emph{Recursive} Markov Chains
(RMCs) can be cast in terms of GPL model checking~\cite{gpl,pip}.

In this paper, we propose an extension to GPL, which we call
\emph{Extended Probabilistic Logic} (XPL), to express properties of
probabilistic systems with \emph{internal nondeterministic choice},
under linear-time semantics.  Syntactically, XPL is very close to GPL:
whereas GPL has probabilistic quantifiers $\mathsf{Pr}_{>p} \psi$ and
$\mathsf{Pr}_{\geq p} \psi$ over fuzzy formulae $\psi$, XPL admits
quantifiers $\mathsf{Pr}_{<p} \psi$ and $\mathsf{Pr}_{\leq p} \psi$ as
well.  XPL's semantics, however, is given with respect to maximizing
schedulers that resolve internal non-deterministic choices. Properties
involving minimizing schedulers can be analyzed by considering their
duals (with respect to negation) over maximizing schedulers.  The
semantics of XPL is defined over Probabilistic Labeled Transition
Systems (PLTSs).  In a PLTS, each state has a set of outgoing
transitions, \emph{possibly with common labels}; and each transition
specifies a distribution of target states.  PLTSs, as interpreted with
XPL, thus exhibit probabilistic choice and, under both linear- and
branching-time semantics, nondeterministic choice.

\medskip
\noindent
\textbf{Contributions and Significance:}\quad XPL is expressive enough
that a wide variety of independently-studied verification problems can
be cast as model checking PLTSs with XPL.  In fact, undecidable
problems such as termination of multi-exit \emph{Recursive} Markov
Decision Processes (Recursive MDPs or RMDPs) can be reduced in linear
time to model checking PLTSs with XPL.  We introduce a
syntactically-defined subclass, called \emph{separable XPL}, for which
model checking is decidable.  We describe a procedure for model
checking XPL which always terminates--- successfully with the model
checking result, or with failure--- such that it always terminates
successfully for separable XPL (see Sect.~\ref{sec:modchk}).

A number of distinct model checking algorithms have been developed
independently for decidable verification problems involving systems
that have probabilistic and internal non-deterministic choice.
Examples of such problems include PCTL* model checking of
MDPs~\cite{PCTL*}, reachability in branching MDPs~\cite{bmdp}, and
termination of 1-exit RMDPs~\cite{rmdp}.  These problems can all be
reduced, in linear time, to model checking \emph{separable} XPL
formulae over PLTSs (see Sect.~\ref{sec:encex}).  To the best of our
knowledge, the idea that branching and recursive systems could be
interpreted as having nondeterminism under the branching-time
semantics, and the question of its compatibility with nondeterminism
under the linear-time semantics, have not been recognized in the
literature.

Termination of multi-exit RMDPs, cast as a model checking problem over
XPL along the same lines as our treatment of 1-exit RMDPs, yields an
XPL formula that is not separable.  Thus separability can be seen as a
characteristic of the verification problems that are known to be
decidable, when cast in terms of model checking in XPL.  Consequently,
XPL in general, and separable XPL in particular, form a useful
formalism to study the relationships between verification problems
over systems involving probabilistic and both linear- and
branching-time non-deterministic choice.  We discuss these issues in
greater detail in Sect.~\ref{sec:concl}.

%---------------------------------------------------------------------
\section{Preliminaries}\label{sec:gpl}
In this section, we formally define PLTSs, which are used to define
the semantics of XPL.  We also summarize the syntax and semantics of
GPL, using the notations from~\cite{gpl}.

\subsection{Probabilistic Labeled Transition Systems}\label{sec:rplts}
We define a probabilistic labeled transition system
(PLTS) as an extension of \cite{gpl}'s~RPLTS.
\begin{definition}[PLTS]\label{def:nlts}
  With respect to fixed sets $Act$ and $Prop$ of \emph{actions} and
  \emph{propositions}, respectively, a PLTS $L$ is a quadruple $(S,
  \delta, P, I)$, where
  \begin{itemize}
  \item $S$ is a countable set of states;

  \item $\delta \subseteq S \times Act \times S$ is the transition
    relation;

  \item $P \fcn{\delta \times \bbbn} [0,1]$ is the transition
    probability distribution satisfying:
    \begin{itemize}
    \item $\forall s \in S . \forall a \in Act . \forall c \in \bbbn.
      \sum\limits_{s':(s,a,s') \in \delta} P(s,a,s',c) \in \{0,1\}$, and

    \item $\forall s \in S . \forall a \in Act . \forall s' \in S
      . (s,a,s') \in \delta \implies (\exists c \in \bbbn
      . P(s,a,s',c) > 0)$;
    \end{itemize}
  \item $I \fcn{S} 2^{Prop}$ is the \emph{interpretation}, recording
    the set of propositions true at a state.
  \end{itemize}
\end{definition}
A reactive PLTS does not have internal nondeterminism, i.e., its
transition probability distribution $P$ is a function of $\delta$.
This definition is in line with the most general for a
PLTS~\cite{indp, proba}, in which, given an action, a probabilistic
distribution is chosen nondeterministically (we assume that there are
finitely many nondeterministic choices).  Other equally expressive
models include alternating automata, in which labeled nondeterministic
ones are followed by silent probabilistic choices.  The difference
between such models has been analyzed with respect to
bisimulation~\cite{bisim}.

Given $L = (S, \delta, P, I)$, a \emph{partial computation} is a
sequence $\sigma = s_0 \tarrow{a_1} s_1 \tarrow{a_2} \cdots
\tarrow{a_n} s_n$, where for all $0 \leq i < n$, $(s_i, a_{i+1},
s_{i+1}) \in \delta$.  Also, $\mathsf{fst}(\sigma) = s_0$ and
$\mathsf{last}(\sigma) = s_n$ denote, respectively, the first and last
states in $\sigma$.  Each transition of a partial computation is
labeled with an action $a_i \in Act$.  The set of all partial
computations of $L$ is denoted by $\mathcal{C}_L$, and
$\mathcal{C}_L(s) = \{\sigma \in \mathcal{C}_L \mid
\mathsf{fst}(\sigma) = s\}$.  \emph{Composition} of partial
computations, $\sigma \tarrow{a} \sigma'$, represents $s_0
\tarrow{a_1} \cdots \tarrow{a_n} s_n \tarrow{a} s'_0 \tarrow{b_1}
\cdots \tarrow{b_m} s'_m$ if $(s_n, a, s'_0) \in \delta$.
A partial computation $\sigma'$ is a \emph{prefix} of $\sigma$ if
$\sigma' = s_0 \tarrow{a_1} \cdots \tarrow{a_i} s_i$ for some $i \leq
n$.

From a set of partial computations, we can build deterministic trees
\emph{(d-trees)}.  We often denote a d-tree by the set of paths in the
tree.  Every d-tree is prefix-closed and deterministic.  $T \subseteq
\mathcal{C}_L$ is \emph{prefix-closed} if, for every $\sigma \in T$
and $\sigma'$ a prefix of $\sigma$, $\sigma' \in T$.  $T$ is
\emph{deterministic} if for every $\sigma, \sigma' \in T$ with $\sigma
= s_o \tarrow{a_1} \cdots \tarrow{a_n} s_n \tarrow{a} s \cdots$ and
$\sigma' = s_0 \tarrow{a_1} \cdots \tarrow{a_n} s_n \tarrow{a'} s'
\cdots$, either $a \neq a'$ or $s = s'$, i.e., if a pair of
computations share a prefix, the first difference cannot involve
transitions labeled by the same action.  A d-tree $T$ has a starting
state, denoted $\mathsf{root}(T)$; if $s = \mathsf{root}(T)$ then $T
\subseteq \mathcal{C}_L(s)$.  We also let $\mathsf{edges}(T) =
\{(\sigma, a, \sigma') \mid \sigma, \sigma' \in T \land \exists s \in
S . \sigma' = \sigma \tarrow{a} s\}$.

$\mathcal{T}_L$ refers to all the d-trees of $L$, and
$\mathcal{T}_L(s) = \{T \in \mathcal{T}_L \mid \mathsf{root}(T) =
s\}$.  $T'$ is a \emph{prefix} of $T$ if $T' \subseteq T$.  $T
\tarrow{a} T'$ means $T' = \{\sigma \mid \mathsf{root}(T) \tarrow{a}
\sigma \in T\}$.  $T$ is \emph{finite} if $|T| < \infty$, and
\emph{maximal} if there exists no d-tree $T'$ with $T \subset T'$.
$\mathcal{M}_L$ and $\mathcal{M}_L(s)$ are analogous to
$\mathcal{T}_L$ and $\mathcal{T}_L(s)$, but for maximal d-trees.  An
\emph{outcome} is a maximal d-tree.

\begin{figure}
  \centering
  \subfloat[An example PLTS]{\label{fig:plts}
  \begin{adjustbox}{scale=1}
    \begin{tikzpicture}[shorten >=1pt,%
      inner sep = 0.5mm, %
      node distance=1.5cm,%
      every state/.style={circle, inner sep=3pt, minimum size=2mm}, %
      on grid,auto,%
      initial text=,%
      highlight/.style={very thick, draw=red, text=red},%
      bend angle=30]
      \node[state] (s_1) {$s_1$};
      \node[state] (s_2) [below=1.25cm of s_1] {$s_2$};
      \node[state] (s_3) [left=1.75cm of s_2] {$s_3$};
      \node[state] (s_4) [right=1.75cm of s_2] {$s_4$};
      \node[state] (s_5) [below=of s_3] {$s_5$};
      \node[state] (s_6) [below=of s_4] {$s_6$};

      \path[->] (s_1) edge [near start] node {$a$} (s_2)
      (s_2) edge [bend right, swap] node {$b,c$} (s_3)
      (s_2) edge [bend left] node {$b,c$} (s_4)
      (s_3) edge [bend right, swap] node {$a:\frac{2}{3}$} (s_2)
      (s_3) edge [swap] node {$a:\frac{1}{3}$} (s_5) 
      (s_4) edge [bend left] node {$a:\frac{3}{4}$} (s_2) 
      (s_4) edge node {$a:\frac{1}{4}$} (s_6);
    \end{tikzpicture}
  \end{adjustbox}
  }
\\  
  \begin{adjustbox}{scale=0.9}
    \begin{tikzpicture}[shorten >=1pt,%
      inner sep = 0.5mm, %
      node distance=1.5cm,%
      every state/.style={circle, inner sep=3pt, minimum size=2mm}, %
      on grid,auto,%
      initial text=,%
      highlight/.style={very thick, draw=red, text=red},%
      bend angle=30]
      \node[state] (s_1) {$s_1$};
      \node[state] (s_2) [below=1.25cm of s_1] {$s_2$};
      \node[state] (s_3) [left=1.25cm of s_2] {$s_3$};
      \node[state] (s_4) [right=1.25cm of s_2] {$s_4$};
      \node[state] (s_5) [below=of s_3] {$s_5$};
      \node[state] (s_6) [below=of s_4] {$s_6$};

      \path[->] (s_1) edge [near start] node {$a$} (s_2)
      (s_2) edge [swap] node {$b$} (s_3)
      (s_2) edge node {$c$} (s_4)
      (s_3) edge [swap] node {$a:\frac{1}{3}$} (s_5) 
      (s_4) edge node {$a:\frac{1}{4}$} (s_6);
    \end{tikzpicture}
  \end{adjustbox}
%} 
\\[1em]

  \subfloat[Example outcomes]{\label{fig:dtrees}
    \begin{adjustbox}{scale=0.9}
    \begin{tikzpicture}[shorten >=1pt,%
      inner sep = 0.5mm, %
      node distance=1.5cm,%
      every state/.style={circle, inner sep=3pt, minimum size=2mm}, %
      on grid,auto,%
      initial text=,%
      highlight/.style={very thick, draw=red, text=red},%
      bend angle=30]
      \node[state] (s_1) {$s_1$};
      \node[state] (s_2) [below=1.25cm of s_1] {$s_2$};
      \node[state] (s_3) [left=1.25cm of s_2] {$s_3$};
      \node[state] (s_4) [right=1.25cm of s_2] {$s_4$};
      \node[state] (s_5) [below=of s_3] {$s_5$};
      \node[state] (s_21) [below=of s_4] {$s_2$};
      \node[state] (s_31) [left=1.25cm of s_21] {$s_3$};
      \node[state] (s_41) [right=1.25cm of s_21] {$s_4$};
      \node[state] (s_51) [below=of s_31] {$s_5$};
      \node[state] (s_6) [below=of s_41] {$s_6$};

      \path[->] (s_1) edge [near start] node {$a$} (s_2)
      (s_2) edge [swap] node {$b$} (s_3)
      (s_2) edge node {$c$} (s_4)
      (s_3) edge [swap] node {$a:\frac{1}{3}$} (s_5) 
      (s_4) edge node {$a:\frac{3}{4}$} (s_21)
      (s_21) edge [swap] node {$c$} (s_31)
      (s_21) edge node {$b$} (s_41)
      (s_31) edge [swap] node {$a:\frac{1}{3}$} (s_51) 
      (s_41) edge node {$a:\frac{1}{4}$} (s_6);
    \end{tikzpicture}
  \end{adjustbox}
  }
  \caption{Example PLTS and selected outcomes}
  \label{fig:example}
\end{figure}

An example PTLS and two of its outcomes are shown in
Fig.~\ref{fig:example}.  In the figure, transitions are usually
annotated with their action label and probability; the probability is
omitted when it is $1$.  Note that there are two transitions labeled
$b$ from state $s_2$ reflecting internal nondeterminism.  If we label
transition from $s_2$ to $s_3$ only with $b$ (omitting $c$) and that
from $s_2$ to $s_4$ only with $c$ (omitting $b$), we get an
\emph{RPLTS} with only probabilistic and external choices.

Note that, with d-trees, we have the distinction between linear- and
branching-time semantics for the nondeterministic choices which are
internal and external, respectively.  Since a d-tree is defined to be
\emph{deterministic}, all of the internal 
choices (both probabilistic and nondeterministic) are resolved, but
the external choices remain.  
Meanwhile, a property of a PLTS will
hold for some subset of its maximal d-trees.  In order to give the
property a probability, we need a measure of this set.  This is
straightforward for an RPLTS, as all internal choices are
probabilistic; but we will need to do more for PLTSs with internal
nondeterministic choice.

Thus, the subsequent concepts apply \emph{only to RPLTSs}, and we will
extend them to PLTSs in Sect.~\ref{sec:xpl}.  A finite RPLTS d-tree
has finite measure, which can be computed from the values of the
probabilistic choices in the trees, i.e., its edges.  An infinite
d-tree will typically have zero measure, but an infinite set of these
may have positive measure.  Instead, intuitively, we consider the
probability of some finite prefix, which again is the product of the
probabilities of all the edges.  Formally, a \emph{basic cylindrical
  subset} of $\mathcal{M}_L(s)$ contains all trees sharing a given
prefix.  Letting $s \in S$, and $T \in \mathcal{T}_L(s)$ to be finite,
$B_T = \{T' \in \mathcal{M}_L \mid T \subseteq T'\}$.  The measure of
$B_T$ is:
\begin{equation}
  \mathsf{m}(B_T) = \prod_{(\sigma, a, \sigma') \in \mathsf{edges}(T)}
  P(\mathsf{last}(\sigma), a, \mathsf{last}(\sigma'))
\end{equation}
From here, a probability measure $\mathsf{m}_s \fcn{\mathcal{B}_s}
[0,1]$ on the smallest field of sets $\mathcal{B}_s$ is generated from
subsets $B_T$ with $\mathsf{m}_s(B_T) =
\mathsf{m}(B_T)$~\cite[Definition~8]{gpl}.

\subsection{GPL Syntax}\label{sec:gsyn}
GPL has two different kinds of formulae.  State formulae depend
directly only on the given state.  \emph{Fuzzy formulae} depend on
\emph{outcomes}.  We give the syntax of GPL, with $X \in Var$, $a \in
Act$, $A \in Prop$, and $0 \leq p \leq 1$, for state formulae, $\phi$,
and fuzzy formulae, $\psi$, as:
\begin{equation*}
  \begin{array}{l}
    \phi ::= A \mid {\neg A} \mid {\phi \land \phi} \mid {\phi \lor
      \phi} \mid \mathsf{Pr}_{>p} \psi \mid \mathsf{Pr}_{\geq p} \psi
    \\

    \psi ::= \phi \mid X \mid {\psi \land \psi} \mid {\psi \lor
      \psi} \mid {\diam{a} \psi} \mid {[a] \psi} \mid {\mu X.\psi}
    \mid {\nu X.\psi}
  \end{array}
\end{equation*}
Note that only atomic propositions may be negated, but every operator
has its dual given in the syntax.  The propositional connectives,
$\land$ and $\lor$, can be used on both state and fuzzy formulae.
Operators $\mu X.\psi$ and $\nu X.\psi$ are least and greatest fixed
point operators for the ``equation'' $X = \psi$.  Additionally, fuzzy
formulae must be alternation-free, which prohibits a kind of mixing of
least and greatest fixed points, and a formula $\psi$ used to
construct state formulae $\mathsf{Pr}_{>p}\psi$ and $\mathsf{Pr}_{\geq
  p}\psi$ may not have any free variables.  These operators check the
probability for a fuzzy formula $\psi$ ($\mathsf{Pr}_{>p}$ and
$\mathsf{Pr}_{\geq 1-p}$ are duals).  The semantics of GPL is given in
terms of RPLTS d-trees.  In that interpretation, \emph{diamond}
implies \emph{box}: $\diam{a} \psi$ means that there is an
$a$-transition and it satisfies $\psi$; $[a] \psi$ means that if there
is an $a$-transition, it satisfies $\psi$.  We also use a set $\alpha
\subseteq Act$ for the modalities, reading $\diam{\alpha} \psi$ as
$\bigvee\limits_{a \in \alpha} \diam{a} \psi$ and $[\alpha] \psi$ as
$\bigwedge\limits_{a \in \alpha} [a] \psi$.  When we write ``$-$'' for
$\alpha$, that represents $Act$.

\begin{table}
  \caption{GPL/XPL semantics: fuzzy formulae}\label{tab:gff}
  \vspace{0.1in}
  \centering
  \begin{tabular}{rl}
    \hline

    $\Theta_L(\phi)e$ & $= \bigcup\limits_{s \models_L \phi}
    \mathcal{M}_L(s)$, where $\phi$ is a closed formula, \\

    $\Theta_L(X)e$ & $= e(X)$, \\

    $\Theta_L(\diam{a} \psi)e$ & $= \{T \in \mathcal{M}_L \mid \exists
    T' : T \tarrow{a} T' \land T' \in \Theta_L(\psi)e\}$, \\

    $\Theta_L([a] \psi)e$ & $= \{T \in \mathcal{M}_L \mid (T
    \tarrow{a} T') \Rightarrow T' \in \Theta_L(\psi)e\}$, \\

    $\Theta_L(\psi_1 \land \psi_2)e$ & $= \Theta_L(\psi_1)e \cap
    \Theta_L(\psi_2)e$, \\

    $\Theta_L(\psi_1 \lor \psi_2)e$ & $= \Theta_L(\psi_1)e \cup
    \Theta_L(\psi_2)e$, \\

    $\Theta_L(\mu X.\psi)e$ & $= \bigcup\limits_{i=0}^\infty M_i$,
    where $M_0 = \emptyset$ and $M_{i+1} = \Theta_L(\psi)e[X \mapsto
      M_i]$, \\

    $\Theta_L(\nu X.\psi)e$ & $= \bigcap\limits_{i=0}^\infty N_i$,
    where $N_0 = \mathcal{M}_L \mbox{ and } N_{i+1} =
    \Theta_L(\psi)e[X \mapsto N_i]$. \\ \hline
  \end{tabular}
\end{table}

\subsection{GPL Semantics}\label{sec:gsem}
\begin{table}
  \caption{GPL semantics: state formulae}\label{tab:gsf}
  \vspace{0.1in}
  \centering
  \begin{tabular}{ll}
    \hline

    $s \models_L A$ &iff $A \in I(s)$, \\

    $s \models_L \neg A$ &iff $A \notin I(s)$, \\

    $s \models_L \phi_1 \land \phi_2$ &iff $s \models_L \phi_1$ and $s
    \models_L \phi_2$, \\

    $s \models_L \phi_1 \lor \phi_2$ &iff $s \models_L \phi_1$ or $s
    \models_L \phi_2$, \\

    $s \models_L \Pr_{>p} \psi$ &iff $\mathsf{m}_s(\Theta_{L,s}(\psi))
    > p$, \\

    $s \models_L \Pr_{\geq p} \psi$ &iff
    $\mathsf{m}_s(\Theta_{L,s}(\psi)) \geq p$. \\

    \hline
  \end{tabular}
\end{table}
We define the semantics of GPL with respect to a fixed RPLTS $L = (S,
\delta, P, I)$, where $\Phi$ and $\Psi$ are the sets of all state and
fuzzy formulae, respectively.  A function $\Theta_L \fcn{\Psi}
2^{\mathcal{M}_L}$, augmented with an extra \emph{environment}
parameter $e \fcn{Var} 2^{\mathcal{M}_L}$, returns the set of outcomes
satisfying a given fuzzy formula, defined inductively in
Table~\ref{tab:gff}.

For a given $s \in S$, $\Theta_{L,s}(\psi) = \Theta_L(\psi) \cap
\mathcal{M}_L(s)$.  The relation $\models_L \subseteq S \times \Phi$
indicates when a state satisfies a state formula, and it is defined
inductively in Table~\ref{tab:gsf}.  Note that the definitions for
$\Theta_L$ and $\models_L$ are mutually recursive.

There are two properties of GPL fuzzy formulae that are important for
the completeness of the GPL model checking algorithm.  First, we have
distributivity on $\emph{box}$ and $\emph{diamond}$
\cite[Lemma~1]{gpl}:
\begin{lemma}[Distributivity on modal operators]
  \label{lem:distr}
  Letting $\oplus \in \{\land, \lor\}$:
  \begin{equation}
    \begin{array}{rcl}
    \Theta_L([a] \psi_1 \oplus [a] \psi_2) &=& \Theta_L\big([a](\psi_1
    \oplus \psi_2)\big) \\

    \Theta_L(\diam{a} \psi_1 \oplus \diam{a} \psi_2) &=&
    \Theta_L\big(\diam{a}(\psi_1 \oplus \psi_2)\big) \\

    \Theta_L([a] \psi_1 \land \diam{a} \psi_2) &=&
    \Theta_L\big(\diam{a}(\psi_1 \land \psi_2)\big)
    \end{array}
  \end{equation}
\end{lemma}

Second, we can relate the probability of a conjunction with that of a
disjunction and compute the effect of taking a step
\cite[Lemma~$2$]{gpl}:
\begin{multline}
    \label{eqn:gl21}
    \mathsf{m}_s(\Theta_{L,s}(\psi_1 \lor \psi_2)) =
    \mathsf{m}_s(\Theta_{L,s}(\psi_1)) +
    \mathsf{m}_s(\Theta_{L,s}(\psi_2)) \: - \\ - \:
    \mathsf{m}_s(\Theta_{L,s}(\psi_1 \land \psi_2))
\end{multline}
\begin{equation}
  \label{eqn:gl22}
  \mathsf{m}_s(\Theta_{L,s}(\diam{a} \psi)) = \sum_{s':(s,a,s') \in
    \delta} P(s,a,s') \cdot \mathsf{m}_{s'}(\Theta_{L,s'}(\psi))
\end{equation}
Additionally, although there is no negation operator in the syntax, we
can write the negation of a fuzzy formula $\psi$,
$\mathsf{neg}(\psi)$, and of a state formula $\phi$,
$\mathsf{neg}(\phi)$, such that, for any RPLTS $L$ and state $s$
(\cite[Lemma~$3$]{gpl}):
\begin{equation*}
  \Theta_{L,s}(\mathsf{neg}(\psi)) = \mathcal{M}_L(s) -
  \Theta_{L,s}(\psi) \quad \mbox{and} \quad \models_L
  \mathsf{neg}(\phi) \iff s \not\models_L \phi \enspace .
\end{equation*}
The proof involves switching all the operators to their duals.

%---------------------------------------------------------------------
\section{XPL}\label{sec:xpl}
To resolve the nondeterministic transitions in a PLTS, we additionally
require a scheduler.  Recall, from Sect.~\ref{sec:rplts}, that
$\mathcal{C}_L$ is the set of all partial computations $\sigma$ of
$L$.
\begin{definition}[Scheduler]\label{def:sched}
  A scheduler for a PLTS $L$ is a function $\gamma \fcn{\mathcal{C}_L
    \times Act} \bbbn$, such that if an action $a$ is present at $s =
  \kw{last}(\sigma)$, then $\gamma(\sigma, a) = c$ implies that
  $\sum_{s'} P(s,a,s',c) = 1$.
\end{definition}
Note that we have defined deterministic schedulers, which are also
aware of their relevant histories.  Given a scheduler $\gamma$ for a
PLTS $L$, we have a (countable) RPLTS $L_\gamma$, where $S_{L,\gamma}
\subseteq \mathcal{C}_L$ and so $\delta_{L,\gamma} \subseteq
\mathcal{C}_L \times Act \times \mathcal{C}_L$.  We define a
probability distribution:
\begin{definition}[Combined probability]\label{def:cmpr}
  The probability distribution of a PLTS $L$ with scheduler $\gamma$
  is a function, $P_{L,\gamma} \fcn{\delta_{L,\gamma}}
  [0,1]$, where:
  \begin{equation}
    P_{L,\gamma}(\sigma, a, \sigma') = P_L(\mathsf{last}(\sigma), a,
    \mathsf{last}(\sigma'), \gamma(\sigma, a))
  \end{equation}
\end{definition}
We also let $P_{L,\gamma}(\sigma, a, \sigma') = 0$ when $(\sigma, a,
\sigma') \notin \delta_{L,\gamma}$.

Recall, from Sect.~\ref{sec:rplts}, that the basic cylindrical
subset $B_T$ contains all maximal d-trees sharing the prefix tree $T$.
For these subsets, we define the probability measure:
\begin{definition}[Probability measure]
  For a PLTS $L$ with scheduler $\gamma$, the probability measure of a
  basic cylindrical subset $B_T$ is defined by a partial function
  $\mathsf{m}^\gamma \fcn{2^{\mathcal{M}_L}} [0,1]$, where:
  \begin{equation}
    \mathsf{m}^\gamma(B_T) = \prod\limits_{(\sigma, a, \sigma') \in
      \mathsf{edges}(T)} P_{L,\gamma}(\sigma, a, \sigma')
  \end{equation}
\end{definition}
Since $\mathsf{m}^\gamma$ may be considered as defined for an RPLTS,
we can extend it to a measure $\mathsf{m}^{\gamma}_s$ as in
Sect.~\ref{sec:rplts}.

\subsection{XPL Syntax}\label{sec:esyn}
Now we give the XPL syntax, with $\mathord{\bowtie} \in \{>, \geq, <,
\leq\}$:
\begin{equation*}
  \begin{array}{l}
    
    \phi ::= A \mid {\neg A} \mid {\phi \land \phi} \mid {\phi \lor
      \phi} \mid {\mathsf{Pr}_{\bowtie p} \psi} \\

    \psi ::= \phi \mid X \mid {\psi \land \psi} \mid {\psi \lor
      \psi} \mid {\diam{a} \psi} \mid {[a] \psi} \mid {\mu X.\psi}
    \mid {\nu X.\psi}
  \end{array}
\end{equation*}
The fuzzy formulae remain the same as in GPL.  $\mathsf{Pr}$ assumes
maximizing schedulers, i.e., we compare against the supremum
probabilities over all schedulers.  Note that $\mathsf{Pr}_{> p}$ is
no longer the dual of $\mathsf{Pr}_{\geq 1-p}$, which is why we allow
the ``less than'' comparisons, as well; moreover, analyzing a fuzzy
formula $\psi$ over minimizing schedulers is essentially equivalent to
considering $\kw{neg}(\psi)$ over maximizing schedulers.

\subsection{XPL Semantics}\label{sec:esem}
\begin{table}
  \caption{XPL semantics: state formulae}\label{tab:esf}
  \vspace{0.1in}
  \centering
  \begin{tabular}{ll}
    \hline

    $s \models_L A$ &iff $A \in I(s)$, \\
    
    $s \models_L \neg A$ &iff $A \notin I(s)$, \\

    $s \models_L \phi_1 \land \phi_2$ &iff $s \models_L \phi_1$ and $s
    \models_L \phi_2$, \\

    $s \models_L \phi_1 \lor \phi_2$ &iff $s \models_L \phi_1$ or $s
    \models_L \phi_2$, \\

    $s \models_L \mathsf{Pr}_{\bowtie p} \psi$ &iff $\sup_\gamma
    \mathsf{m}_s^\gamma(\Theta_{L,s}(\psi)) \bowtie p$, \\

    \hline
  \end{tabular}
\end{table}
The semantics of XPL changes from GPL only due to the measure of the
PLTS outcomes.  In particular, we retain the same semantics on
\emph{diamond} and \emph{box}.  The semantics is defined with respect
to a fixed PLTS $L = (S, \delta, P, I)$.  The function $\Theta_L
\fcn{\Psi} 2^{\mathcal{M}_L}$ remains the same, while $\models_L
\subseteq S \times \Phi$ differs for the probabilistic operators.

\begin{definition}[XPL semantics]
  The semantics for the state formulae is given in
  Table~\ref{tab:esf}.  For the fuzzy formulae, the semantics are as
  in Table~\ref{tab:gff}.
\end{definition}

Note the use of $\sup$ and $\inf$ in Table~\ref{tab:esf}.  We refer to
the value $\sup_\gamma \mathsf{m}_s^\gamma(\Theta_{L,s}(\psi))$ as a
\emph{probabilistic value} and write it as
$\mathsf{Pr}_{L,s}(\psi)$ (\cite{weakb2} calls this a
\emph{capacity}).  Unlike in GPL, we may not always be able to compute
it with a model checking algorithm.

\subsection{Separability of Fuzzy Formulae}\label{sec:sep}
With internal nondeterminism, we lose the general relation between
conjunctions and disjunctions, as in (\ref{eqn:gl21}).  However, since
we are maximizing (or minimizing) over schedulers, we would want the
relation in (\ref{eqn:ndf}).
\begin{equation}\label{eqn:ndf}
  \mathsf{Pr}_{L,s}(\psi_1 \lor \psi_2) \stackrel{?}{=}
  \mathsf{Pr}_{L,s}(\psi_1) + \mathsf{Pr}_{L,s}(\psi_2) -
  \mathsf{Pr}_{L,s}(\psi_1 \land \psi_2)
\end{equation}

This requires that the optimal strategy be the same for $\psi_1$,
$\psi_2$, $\psi_1 \land \psi_2$, and $\psi_1 \lor \psi_2$; in general,
these may all be distinct.  Instead, we will seek to delay the
application of all conjunctions and disjunctions until the two sides
are \emph{independent}, primarily through repeated application of
Lemma~\ref{lem:distr}, which holds for XPL as well because it deals
with sets of d-trees, but not their measure.  For example, we can
rewrite $\psi_a = \nu X.\diam{a}\diam{b}X \lor \diam{a}\diam{c}X$ as
$\nu X.\diam{a}(\diam{b}X \lor \diam{c}X)$.
We generalize this to a syntactic notion of \emph{separability},
defined below.  It will be useful to view a fuzzy formula as a kind of
an \emph{and-or tree}.

\begin{definition}[And-or tree]\label{def:and-or}
  The and-or tree of a fuzzy formula $\psi$, $AO(\psi)$ is a node
  labeled by $\oplus$, where $\oplus \in \{\land, \lor\}$, with
  children $AO(\psi_1)$ and $AO(\psi_2)$ when $\psi = \psi_1 \oplus
  \psi_2$, and a leaf $\psi$ otherwise.
\end{definition}

We can flatten this tree with the straightforward flattening operator,
where, e.g., the tree $\land(\psi_1, \dots, \land(\psi_2, \psi_3))$
may be flattened to $\land(\psi_1, \dots, \psi_2, \psi_3)$.  Note that
flattened trees have alternating $\land$ and $\lor$ nodes.  A
(conjunctive) set of formulae $F$ corresponds to a flattened and-or
tree with the root node labeled by $\land$ and having the elements of
$F$ as leaves.  We will assume $AO(\psi)$ refers to the flattened
tree.

A subformula of $\psi$ of the form $\diam{a}\psi'$ or $[a]\psi'$ is
called a \emph{modal subformula} of $\psi$.  We say that $\psi'$ is an
\emph{unguarded subformula} of $\psi$
if it is a leaf in $AO(\psi)$.
The GPL model checking algorithm requires bound variables to be
guarded by actions (i.e., $\mu X.([a] X \land \dots)$ is fine, but
$\mu X.(X \land \dots)$ is not)~\cite{gpl}, and we adopt this
requirement as well.

\begin{definition}[Formula Transformations]\label{def:formtrans}
\ 
  \begin{itemize}
  \item The \emph{fixed-point expansion} of $\psi$, denoted by
    $\fpe(\psi)$, is a formula $\psi'$ obtained by expanding any
    unguarded subformula of the form $\sigma X.\psi_X$ to
    $\psi_X[\sigma X.\psi_X/X]$ where $\sigma \in \{\mu, \nu\}$.

  \item We say that a formula is non-probabilistic if it is a state
    formula, or of the form $\diam{a} \phi$ and $[a] \phi$ for $a \in
    Act$ and $\phi \in \{\mathsf{tt}, \mathsf{ff}\}$.  The
    \emph{purely probabilistic abstraction} of a fuzzy formula $\psi$,
    denoted by $\probabs(\psi)$, is a formula obtained by removing
    unguarded non-probabilistic subformulae (i.e., $\psi' \land \phi$,
    where $\phi$ is non-probabilistic, becomes $\psi'$, etc.).

  \item A \emph{grouping} of a formula $\psi$, denoted by
    $\groupmodalities(\psi)$, groups modalities in a formula using
    distributivity. Formally, $\groupmodalities$ maps $\psi$ to a
    $\psi'$ that is equivalent to $\psi$ based on the equivalences in
    Lemma~\ref{lem:distr}, applied left-to-right as much as possible
    on the top level.
  \end{itemize}  
\end{definition}

At a high level, a necessary condition of separability is that the
actions guarding distinct conjuncts and disjuncts of a formula are
distinct as well.

\begin{definition}[Action set]\label{def:act}
  The \emph{action set} of a formula $\psi$, denoted by
  $\mathsf{action}(\psi)$
  is the set of actions appearing at unguarded modal subformulae of
  $\psi$:
  \begin{itemize} 
  \item $\mathsf{action}(\phi) = \emptyset$;

  \item $\mathsf{action}(\diam{a} \psi) = \mathsf{action}([a] \psi) =
    \{a\}$;

  \item $\mathsf{action}(\psi_1 \land \psi_2) = \mathsf{action}(\psi_1
    \lor \psi_2) = \mathsf{action}(\psi_1) \cup
    \mathsf{action}(\psi_2)$;

  \item $\mathsf{action}(\mu X.\psi) = \mathsf{action}(\nu X.\psi) =
    \mathsf{action}(\psi)$.
  \end{itemize}
\end{definition}

We can now define separability based on action sets of formulae
as follows.

\begin{definition}[Separability]\label{def:isep}
  The set of all separable formulae is the largest set $\mathcal{S}$
  such that $\forall \psi \in \mathcal{S}$, if $\psi' =
  \groupmodalities(\probabs(\fpe(\psi)))$, then 
  \begin{enumerate}
  \item every subformula of $\psi'$ is in $\mathcal{S}$, and 
  \item if $\psi' = \psi_1 \oplus \psi_2$ where $\oplus \in \{\land,
    \lor\}$, then $\mathsf{action}(\psi_1) \cap \mathsf{action}(\psi_2) =
    \emptyset$.
  \end{enumerate}
  A formula $\psi$ is \emph{separable} if $\psi \in \mathcal{S}$.
\end{definition}

Below we illustrate separability of formulae.  Let $\psi_1$-$\psi_4$
be all separable and distinct, and also let $\psi_1 \lor \psi_2$ and
$\psi_3 \lor \psi_4$ be separable.

Note that $\groupmodalities$ uses only distributivity of the modal
operators over ``$\land$'' and ``$\lor$'', and not the distributivity
of the boolean operators themselves.  Consequently, a separable
formula may be equivalent to a non-separable formula.  

\begin{example}[Separable formula with equivalent non-separable formula]\label{ex:sepr}
  The formula $\psi_s$ is separable.
  \begin{equation}
    \label{eqn:sepr}
    \psi_s = [a](\psi_1 \lor \psi_2) \land [b](\psi_3 \lor
    \psi_4)
  \end{equation}
  The DNF version of $\psi_s$, $\psi'_s$, is not separable since
  action sets of disjuncts overlap.
  \begin{equation} 
    \label{eqn:dist}
    \psi'_s = ([a]\psi_1 \land [b]\psi_3) \lor ([a]\psi_1 \land
         [b]\psi_4) \lor ([a]\psi_2 \land [b]\psi_3) \lor ([a]\psi_2
         \land [b]\psi_4)
  \end{equation}
\end{example}

This is important because we need the subformulae of a separable
formula to also be separable.

\begin{example}[Non-separable formula]\label{ex:entg}
  The formula $\psi_e$ is a subformula of $\psi'_s$
  (\ref{eqn:dist}), is not separable, and has no equivalent
  separable formula:
  \begin{equation}\label{eqn:entg}
    \psi_e = ([a]\psi_1 \land [b]\psi_4) \lor ([a]\psi_2 \land
        [b]\psi_3)
  \end{equation}
  With $\psi_e$, we need to satisfy $\psi_1$ or $\psi_2$ following an
  $a$ action, and likewise for $\psi_3$ or $\psi_4$ following a $b$
  action.  An equivalent separable formula would thus have to include
  $[a](\psi_1 \lor \psi_2)$ and $[b](\psi_3 \lor \psi_4)$, but this
  would also be satisfied by, e.g., outcomes satisfying only
  $[a]\psi_1 \land [b]\psi_3$.
\end{example}

We say that a formula is \emph{entangled} at a state if it is not
(equivalent to) a separable formula even after considering that
state's specific characteristics.  For instance, $\psi_e$ is entangled
only at states with both $a$ and $b$ actions present.  Even when
considering only states where the actions relevant to entanglement are
present, a formula may be entangled at some states and not at others.

\begin{example}[Entanglement on $a$ and $b$ depends on $c$]\label{ex:ents}
  The formula $\psi_c$ reduces to $\psi_s'$ (\ref{eqn:sepr}) at states
  that have a $c$-transition, and to $\psi_e$ (\ref{eqn:entg})
  otherwise.
  \begin{multline}\label{eqn:ents}
    \psi_c = ([a]\psi_1 \land [b]\psi_3 \land \diam{c}\kw{tt}) \lor
    ([a]\psi_1 \land [b]\psi_4) \lor ~ \\ \lor ([a]\psi_2 \land
        [b]\psi_3) \lor ([a]\psi_2 \land [b]\psi_4 \land
        \diam{c}\kw{tt}).
  \end{multline}
\end{example}

There are also non-separable formulae that nonetheless would not be
entangled at any state of an arbitrary PLTS.
\begin{example}[Never-entangled non-separable formula]\label{ex:nens}
  For the formula $\psi_d$, $\probabs(\psi_d) = \psi_e$, but at any
  state it is equivalent either to $[a]\psi_1 \land [b]\psi_4$ or to
  $[a]\psi_2 \land [b]\psi_3$.
  \begin{equation}\label{eqn:nens}
    \psi_d = ([a]\psi_1 \land [b]\psi_4 \land [c]\kw{ff}) \lor
    ([a]\psi_2 \land [b]\psi_3 \land \diam{c}\kw{tt}).
  \end{equation}
\end{example}

Since $\groupmodalities$ combines modal subformulae with a common
action, we have the following important consequence.
\begin{remark}
All conjunctive formulae and disjunctive formulae are separable.
\end{remark}

%---------------------------------------------------------------------
\section{Model Checking XPL Formulae}\label{sec:modchk}
We outline a model checking procedure for XPL formulae for a
fixed PLTS $L = (S, \delta, P, I)$, along similar lines to the GPL
model checking algorithm in~\cite[Sect.~$4$]{gpl}.  The model
checking procedure succeeds whenever the given formula is separable.

\begin{definition}[Fisher-Ladner closure] Given a formula $\psi$, its 
\emph{Fisher-Ladner closure}, $Cl(\psi)$, is the smallest set such
that the following hold:
  \begin{itemize}
  \item $\psi \in Cl(\psi)$.
  \item If $\psi' \in Cl(\psi)$, then:
    \begin{itemize}
    \item if $\psi' = \psi_1 \land \psi_2$ or $\psi_1 \lor \psi_2$, then
      $\psi_1, \psi_2 \in Cl(\psi)$;

    \item if $\psi' = \diam{a} \psi''$ or $[a] \psi''$ for some $a \in
      Act$, then $\psi'' \in Cl(\psi)$;

    \item if $\psi' = \sigma X.\psi''$, then $\psi''[\sigma X.\psi'' /
      X] \in Cl(\psi)$, with $\sigma$ either $\mu$ or $\nu$.
    \end{itemize}
 \end{itemize}
\end{definition}

Also, we let $\mathcal{AO}(S)$ represent the set of and-or trees with
elements of a set $S$ as leaves.  The core of the model checking
algorithm is the construction of a \emph{dependency graph} $\depgr(s,
\psi)$, to compute $\mathsf{Pr}_{L,s}(\psi)$, such that all the
formulae appearing in the graph will be in the set
$\mathcal{AO}(Cl(\psi))$.  When constructing a dependency graph, in
order to divide a formula by actions, we transform it into a
\emph{factored} form, in a similar manner to checking separability.
If we are unable to transform a formula into a factored form, as can
happen when a formula is non-separable, the graph construction
terminates with failure.

\begin{definition}[Factored form]
  A factored formula $\psi$
  can be trivial, when $\psi \in \{\mathsf{tt}, \mathsf{ff}\}$.
  Otherwise, every leaf of $AO(\psi)$ is in the \emph{action} form,
  $\diam{a} \psi'$, and no action may guard more than one leaf.
\end{definition}

Given a state $s$, a formula $\psi'$ can be transformed into a
semantically equivalent one $\psi''$ that is in factored
form\footnote{We may use the DNF version of $\psi'$ to check for
  equivalence with existing nodes, but not for finding the factored
  form.} as: $\psi'' =
\groupmodalities\left(PE\big(s,\fpe(\psi')\big)\right)$.  $PE(s,
\psi')$ \emph{partially evaluates} $\psi'$, by evaluating unguarded
non-probabilistic subformulae of $\psi'$ as well as all unguarded
modal subformulae with actions absent at state $s$, yielding
$\mathsf{tt}$ or $\mathsf{ff}$ for each, and simplifying the
result.\footnote{After applying $\groupmodalities$, we may have a leaf
  in action form $\diam{a} \psi'_a \notin \mathcal{AO}(Cl(\psi))$.
  Then, we may view an action $a$ as a prefix label on the subtree
  $\psi'_a \in \mathcal{AO}(Cl(\psi))$.}  Then $\big((s,\psi'),
\varepsilon, (s,\psi'')\big) \in E$.

\begin{definition}[Dependency graph]
  The dependency graph for model checking a formula $\psi$ with
  respect to a state $s$ in PLTS $L$, denoted by $\depgr(s, \psi)$, is
  a directed graph $(N,E)$, where \emph{node set} $N \subseteq S
  \times \mathcal{AO}({Cl(\psi)})$, and \emph{edge set} $E \subseteq N
  \times (Act \cup \{\varepsilon, \varepsilon^\land,
  \varepsilon^\lor\}) \times N$; i.e., the edges are labeled from $Act
  \cup \{\varepsilon, \varepsilon^\land, \varepsilon^\lor\}$.  The
  sets $N$ and $E$ are the smallest such that:
  \begin{itemize}
  \item $(s,\psi) \in N$.

  \item If $(s', \psi') \in N$, $\psi'$ is not in factored form: if
    equivalent $\psi''$ in factored form exists, then $(s',\psi'') \in
    N$ and $((s',\psi'),\varepsilon, (s',\psi'')) \in E$.

  \item If $(s', \psi'_1 \oplus \psi'_2) \in N$, then $(s', \psi'_i)
    \in N$ for $i =1,2$.  Moreover, $((s',\psi'_1 \oplus
    \psi'_2),\varepsilon^\oplus, (s',\psi'_i)) \in E$ for $i=1,2$, and
    $\oplus \in \{\land,\lor\}$.

  \item If $(s', \diam{a}\psi') \in N$, then $(s'', \psi') \in N$ for
    each $s''$ such that $(s', a, s'') \in \delta$.  Moreover, $((s',
    \diam{a}\psi'),a, (s'',\psi')) \in E$.
  \end{itemize}
\end{definition}
If $(s',\psi') \in N$ and $\psi'$ has no factored form, then the
dependency graph construction fails.

When we transform $\psi'$ to the factored form $\psi''$, the semantics
does not change, i.e., $\Theta_{L,s'}(\psi') = \Theta_{L,s'}(\psi'')$.
For the factored formulae, standard XPL semantics applies
(Table~\ref{tab:gff}).  Note that we can assume \emph{action} nodes to
be of the form $(s', \diam{a} \psi')$, as the action $a$ must then be
present at state $s'$.  From this semantics, we also get the
relationships for the probabilistic values.  Here, $\prod$ is the
standard product operator, while $\coprod_{i \in I} x_i = 1 - \prod_{i
  \in I} (1-x_i)$.
\begin{lemma}[Probabilistic values]\label{lem:prv}
  Fix $\depgr(s_0, \psi) = (N, E)$.  The probabilistic value
  $\mathsf{Pr}_{L,s}(\psi')$ for a node $(s, \psi')$ is as
  follows:

  \begin{itemize}
  \item $\mathsf{Pr}_{L,s}(\mathsf{ff}) = 0$ and
    $\mathsf{Pr}_{L,s}(\mathsf{tt}) = 1$.

  \item If $(s, \psi')$ is an \emph{and}-node, then: \\
    $\mathsf{Pr}_{L,s}(\psi') = \prod_{((s,\psi'),
    \varepsilon^\land, (s,\psi'_i)) \in E}
    \mathsf{Pr}_{L,s}(\psi'_i)$.

  \item If $(s, \psi')$ is an \emph{or}-node, then: \\
    $\mathsf{Pr}_{L,s}(\psi') = \coprod_{((s,\psi'),
    \varepsilon^\lor, (s,\psi'_i)) \in E}
    \mathsf{Pr}_{L,s}(\psi'_i)$.

  \item If $(s, \psi')$ is an action node, i.e., $\psi' = \diam{a}
    \psi'_a$, then:
    \begin{equation*}
      \mathsf{Pr}_{L,s}(\psi') = \max_{c \in \bbbn}
      \sum_{((s,\psi'), a, (s', \psi'_a)) \in E} P(s,a,s',c) \cdot
      \mathsf{Pr}_{L,s'}(\psi'_a)
    \end{equation*}

  \item The remaining nodes $(s, \psi')$ have a unique successor $(s,
    \psi'')$ with $\mathsf{Pr}_{L,s}(\psi') =
    \mathsf{Pr}_{L,s}(\psi'')$.
  \end{itemize}
\end{lemma}

\begin{proof}
  Most of the cases are straightforward and similar to the GPL model
  checking algorithm~\cite[Lemma~8]{gpl} and a result for two-player
  stochastic parity games~\cite[Theorem~4.22]{indp}.  The
  \emph{and}-node and \emph{or}-node cases have the product and
  coproduct, respectively, due to independence.  We explain the
  \emph{action node} case in more detail.

  The sum over the probabilistic distribution is as in GPL and
  (\ref{eqn:gl22}); we explain the nondeterministic choice.  A PLTS
  scheduler makes a choice for an action given the partial computation
  $\sigma$.  Here, this choice is made based on a formula, $\psi'_a$,
  to be satisfied.  When the initial formula $\psi$ is separable, this
  is well-defined: given $L$, $s$, and $\psi$, the scheduler can
  deduce $\psi'_a$ from $\sigma$, a la traversal of the dependency
  graph.  %\qed
\end{proof}

We note that, although a particular choice may maximize
$\mathsf{Pr}_{L,s}(\psi')$, a scheduler that makes this choice
\emph{every time} is not necessarily optimal.  Indeed, no optimal
scheduler may exist, in which case we would only have
$\epsilon$-optimal schedulers for any $\epsilon > 0$~\cite{bmdp,
  indp}.  The probabilistic value may be predicated on making a
different choice \emph{eventually}.  The formulation in
Lemma~\ref{lem:prv} is consistent with this possibility, and the
existence of ($\epsilon$-)optimal schedulers may be justified through
a common method, called \emph{strategy improvement} or \emph{strategy
  stealing}~\cite{rmdp, indp}.  The intuition is that, in case of a
loop, we can add a choice to succeed immediately with the maximum
probability for the state.  This cannot increase the probability, and
the maximizing scheduler can otherwise be the same, if this choice
does not arise.

\begin{theorem}[Model checking termination]
  The graph construction of $\depgr(s, \psi)$ terminates for any XPL
  formula $\psi$ and PLTS $L$.  Moreover, if $\psi$ is separable, the
  XPL model checking algorithm will complete the construction of the
  dependency graph.
\end{theorem}
\begin{proof}
  $Cl(\psi)$ is finite, so $\mathcal{AO}(Cl(\psi))$ (for DNF versions
  used for equivalence checking) is finite.  The number of actions in
  $L$ and $\psi$ is finite, so the number of factored formulae is
  finite.  This is sufficient to guarantee termination, as we fail
  when we cannot construct a factored formula.  Meanwhile,
  separability of $\psi$ implies that we can construct a factored
  formula from any $\psi' \in \mathcal{AO}(Cl(\psi))$.  %\qed
\end{proof}

Our primary contribution is the completed dependency graph for a
separable formula $\psi$.  For model checking separable XPL formulae,
we show how, given the graph, to compare the probabilistic value of
$\psi$ at a state $s$ against a threshold $p$.  We do this by first
constructing a \emph{system of polynomial max fixed point equations}
from the graph.  Each node $i$ in the dependency graph is associated
with a real-valued variable $x_i$.  Given a set of variables $V$, each
equation in the system is of the form $x_i = e$ where $e$ is
\begin{itemize}
\item a polynomial over $V$ such that the sum of coefficients is $\leq
  1$; or

\item of the form $\max(V')$ where $V' \subseteq V$.
\end{itemize}
Furthermore, the equations form a stratified system, where each
variable $x_i$ can be assigned a stratum $j=\id{stratum}(x_i)$ such
that $x_i$ is defined in terms of only variables of the form $x_k$
such that $\id{stratum}(x_k) \leq \id{stratum}(x_i)$
(cf.~\cite[Def.~9]{tree15}); and variables in the same stratum $j$
fall under the same fixed point.

\begin{theorem}
Given a real value $p$, a system of polynomial max fixed point
equations and a distinguished variable $x$ defined in the system,
whether or not $x \bowtie p$ in its solution is decidable.
\end{theorem}

\begin{proof}
  We write the max polynomial system, $\mathbf{x} = P(\mathbf{x})$, as
  a sentence in the first-order theory of real closed fields, similar
  to~\cite{tree15}.  The additional comparison will be $x_0 \bowtie
  p$.  Along with the equation system, we need to encode fixed points
  and $\max$.

  We can encode $x_i = \max(x_j, x_k)$ as
  (\ref{eqn:fomax})~(cf.~\cite[Section~5]{rmdp}):
  \begin{equation}\label{eqn:fomax}
    x_i \geq x_j \land x_i \geq x_k \land (x_i \leq x_j \lor x_i \leq
    x_k) \enspace .
  \end{equation}
  Meanwhile, letting $V$ be the set of all variables and $I$ a subset
  belonging to some stratum with least fixed point, we can encode the
  fixed point itself as (\ref{eqn:folfp}):
  \begin{equation}\label{eqn:folfp}
    \forall \mathbf{x}'_I . \left(\bigwedge_{i \in I} x'_i =
    P_i(\mathbf{x}'_I, \mathbf{x}_{V \setminus I}) \implies
    \bigwedge_{i \in I} x_i \leq x'_i \right) \enspace .
  \end{equation}
  The stratification of fixed points in the equation system precludes
  a cyclical dependency between a least and a greatest fixed point; a
  greatest fixed point can be encoded similarly.

  The original fixed point equation system, along with the query $x
  \bowtie p$, (\ref{eqn:fomax})-(\ref{eqn:folfp}), and the counterpart
  encoding greatest fixed point, are sentences in a first order theory
  of real closed fields, which is decidable~\cite{tarski}.  Hence the
  decidability of $x \bowtie p$ in the solution to the fixed point
  equations follows. %\qed
\end{proof}

We use the above result to determine whether or not
$\mathsf{Pr}_{L,s}(\psi) \bowtie p$ for a separable XPL formula
$\psi$.  The polynomial fixed point system is derived similarly
to~\cite[Section~4.1.2]{gpl}, with a variable $x_{(s,\psi)}$ for each
node $(s,\psi)$ in the dependency graph $\depgr(s, \psi)$, and
equations based on Lemma~\ref{lem:prv}.
\begin{itemize}
\item If $\psi$ is not in factored form, then $(s,\psi)$ has a unique
  edge labeled by $\varepsilon$ to a node $(s,\psi')$, and
  $x_{(s,\psi)} = x_{(s,\psi')}$.

\item $x_{(s,\mathsf{ff})} = 0$ and $x_{(s,\mathsf{tt})} = 1$.

\item If $(s,\psi)$ is an \emph{and}-node, then $x_{(s,\psi)} =
  \prod\limits_{((s,\psi), \varepsilon^{\land}, (s,\psi_i)) \in E}
  x_{(s,\psi_i)}$.

\item If $(s,\psi)$ is an \emph{or}-node, then $x_{(s,\psi)} =
  \coprod\limits_{((s,\psi), \varepsilon^{\lor}, (s,\psi_i)) \in E}
  x_{(s,\psi_i)}$.

\item If $(s,\psi)$ is an action node and $\psi = \diam{a} \psi_a$,
  then \\ $x_{(s,\psi)} = \max\limits_{c \in \bbbn} \sum\limits_{((s,\psi),
    a, (s',\psi_a)) \in E} P(s, a, s', c) \cdot x_{(s',\psi_a)}$.
\end{itemize}

\begin{theorem}[Correctness]
  The construction of the dependency graph $\depgr(s, \psi)$, when
  $\psi$ is separable, yields a polynomial max fixed point equation
  system, such that the value of $x_{(s,\psi)}$ in its solution is
  $\mathsf{Pr}_{L,s}(\psi)$.
\end{theorem}

\begin{proof}
  The correctness result follows from Lemma~\ref{lem:prv} and the
  semantics of fixed points given by Equation~\ref{eqn:folfp} (and its
  counterpart). %\qed
\end{proof}

Consequently, we have:

\begin{corollary}[Decidability]
  Given a state formula $\varphi$ with separable subformulae, a PLTS
  $L$ and a state $s$ in $L$, whether or not $s \models_L \varphi$ is
  decidable.
\end{corollary}

\medskip

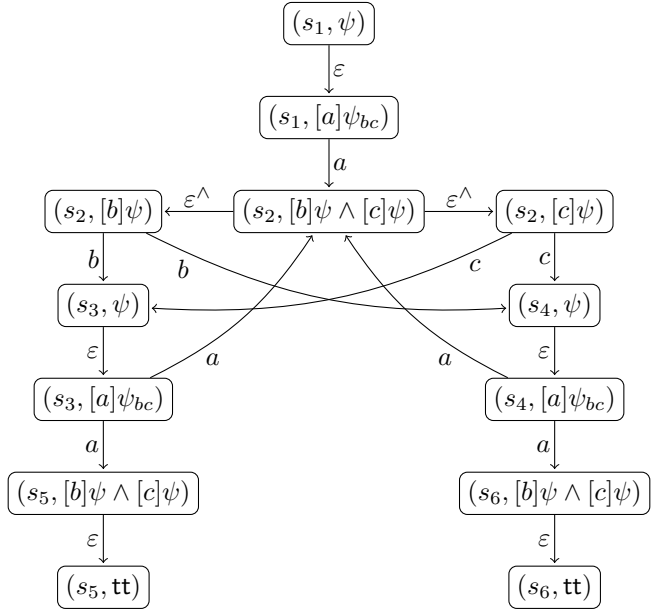
\begin{figure}
  \centering
  \begin{adjustbox}{scale=1}
      \begin{tikzpicture}[shorten >=1pt,%
        inner sep = 0.5mm, %
        node distance=1.25cm,%
        every state/.style={rectangle, rounded corners=3pt, inner sep=3pt, minimum
          size=2mm}, %
        on grid,auto,%
        initial text=,%
        highlight/.style={very thick, draw=red, text=red},%
        bend angle=15]
        \node[state] (s1p) {$(s_1,\psi)$}; 
        \node[state] (s1apbc) [below=of s1p] {$(s_1,[a]\psi_{bc})$}; 
        \node[state] (s2bcp) [below=of s1apbc] {$(s_2,[b]\psi\land [c]\psi)$}; 
        \node[state] (s2bp) [left= 3cm of s2bcp] {$(s_2,[b]\psi)$}; 
        \node[state] (s2cp) [right=3cm of s2bcp] {$(s_2,[c]\psi)$}; 
        \node[state] (s3p) [below=of s2bp] {$(s_3,\psi)$}; 
        \node[state] (s4p) [below=of s2cp] {$(s_4,\psi)$}; 
        \node[state] (s3apbc) [below=of s3p] {$(s_3,[a]\psi_{bc})$}; 
        \node[state] (s4apbc) [below=of s4p] {$(s_4,[a]\psi_{bc})$}; 
        \node[state] (s5bcp) [below=of s3apbc] {$(s_5,[b]\psi\land [c]\psi)$}; 
        \node[state] (s6bcp) [below=of s4apbc] {$(s_6,[b]\psi\land [c]\psi)$}; 
        \node[state] (s5tt) [below=of s5bcp] {$(s_5,\mathsf{tt})$}; 
        \node[state] (s6tt) [below=of s6bcp] {$(s_6,\mathsf{tt})$}; 

        \path[->] 
        (s1p) edge node {$\varepsilon$} (s1apbc) 
        (s1apbc) edge node {$a$} (s2bcp) 
        (s2bcp) edge [swap] node {$\varepsilon^{\land}$} (s2bp) 
        (s2bcp) edge  node {$\varepsilon^{\land}$} (s2cp) 
        (s2bp) edge [swap] node {$b$} (s3p) 
        (s2bp) edge [bend right, swap, very near start] node {$b$} (s4p) 
        (s4p) edge [swap] node {$\varepsilon$} (s4apbc) 
        (s3apbc) edge [bend right, swap, near start] node {$a$} (s2bcp) 
        (s3apbc) edge [swap] node {$a$} (s5bcp) 
        (s5bcp) edge [swap] node {$\varepsilon$} (s5tt) 
        (s2cp) edge [swap] node {$c$} (s4p) 
        (s2cp) edge [bend left, very near start] node {$c$} (s3p) 
        (s3p) edge [swap] node {$\varepsilon$} (s3apbc) 
        (s4apbc) edge [bend left, near start] node {$a$} (s2bcp) 
        (s4apbc) edge [swap] node {$a$} (s6bcp) 
        (s6bcp) edge [swap] node {$\varepsilon$} (s6tt);
      \end{tikzpicture}
  \end{adjustbox}
  \caption{XPL Model Checking Example: Dependency Graph}
  \label{fig:xplp}
\end{figure}
\begin{example}[Model Checking]\label{ex:modchk}
  For the PLTS $L$ in Fig.~\ref{fig:plts} and fuzzy formula $\psi =
  \mu X.[a][b]X \land [a][c]X$, we have symmetric nondeterministic
  choices on $b$ and $c$ from state $s_2$, and the formula is
  satisfied by all finite d-trees (since both $s_3$ and $s_4$ have a
  probability greater than $\frac 1 2$ of returning to $s_2$, the
  infinite d-trees have positive measure on any scheduler).  Letting
  $\psi_{bc} = [b]\psi \land [c]\psi$, we get the dependency graph shown
  in Fig.~\ref{fig:xplp}.

 We find $\mathsf{Pr}_{L,s_1}(\psi)$ as the value of $x^a_1$ in
 the least fixed point from
 the following equations:
 \begin{equation}
   \begin{array}{ll}
     x^a_1 = x^{bc}_2 & x^b_2 = \max(x^a_3, x^a_4) \\

     x^{bc}_2 = x^b_2 \cdot x^c_2 & x^c_2 = \max(x^a_3, x^a_4) \\

     x^a_3 = {\frac 1 3} x^{bc}_5 + {\frac 2 3} x^{bc}_2 \quad ~ &
     x^{bc}_5 = 1 \\

     x^a_4 = {\frac 1 4} x^{bc}_6 + {\frac 3 4} x^{bc}_2 & x^{bc}_6 =
     1
   \end{array}
 \end{equation}
 \noindent
 Solving the equations, we get $\mathsf{Pr}_{L,s_1}(\psi) =
 x^a_1 = \frac 1 4$.
\end{example}

Note that the model checking algorithm can be broken into the
following two parts: writing down a polynomial system, and then
finding the (approximate) solution.  The first part is bounded
double-exponentially in the size of the fuzzy formula, as we deal with
(for complexity purposes on equivalence checking) DNF formulae from
the Fisher-Ladner closure.  For the second part, value iteration is
guaranteed to converge, when computing a single least or greatest
fixed point~\cite{bmdp}, but may be exponentially slow in the number
of digits of precision~\cite{msp}.  For polynomial systems arising
from conjunctive formulae, alternative approximation methods have been
proven to be efficient~\cite{newton}.

%---------------------------------------------------------------------
\section{Encoding Other Model Checking Problems}\label{sec:encex}
\subsection{Model Checking PCTL* over MDPs}\label{sec:pctls}

PCTL* is a widely used and well-known logic for specifying properties
over Markov chains and MDPs.  The syntax of PCTL* may be given as
follows, where $A \in Prop$ and $\phi$ and $\psi$ represent state
formulae and path formulae, respectively:
\begin{equation*}
  \begin{array}{l}
    \phi ::= A \mid \phi \land \phi \mid \neg \phi \mid
    \mathsf{Pr}_{> p} \psi \mid \mathsf{Pr}_{\geq p} \psi \\

    \psi ::= \phi \mid \mathsf{X} \psi \mid \psi \mathsf{U} \psi \mid
    \psi \land \psi \mid \neg \psi
  \end{array}
\end{equation*}
This is similar to the syntax given by~\cite[Chapter~$9$]{PCTL*},
except omitting the bounded \emph{until} operator.

\begin{table}
  \caption{Encoding of PCTL* over MDPs}\label{tab:pctl}
  \begin{equation*}
    E_{PCTL^*}(\gamma) = \left\{
    \begin{array}{ll}
      \hline
      \gamma, & \gamma \in Prop, \\

      \mathsf{neg}(E_{PCTL^*}(\gamma')), & \gamma = \neg \gamma', \\

      E_{PCTL^*}(\gamma_1) \land E_{PCTL^*}(\gamma_2), & \gamma =
      \gamma_1 \land \gamma_2, \\

      \mathsf{Pr}_{> p} E_{PCTL^*}(\psi), & \gamma =
      \mathsf{Pr}_{> p} \psi, \\

      \mathsf{Pr}_{\geq p} E_{PCTL^*}(\psi), & \gamma =
      \mathsf{Pr}_{\geq p} \psi, \\

      \diam{a} E_{PCTL^*}(\psi), & \gamma = \mathsf{X} \psi, \\

      \mu X.E_{PCTL^*}(\psi_2) \lor (E_{PCTL^*}(\psi_1) \land \diam{a}
      X), & \gamma = \psi_1 \mathsf{U} \psi_2. \\ \hline
    \end{array}
    \right.
  \end{equation*}
\end{table}

In \cite[Sect.~3.2]{gpl}, PCTL* model checking over Markov chains was
encoded in terms of GPL model checking over RPLTSs.  First of all,
Markov chains are represented as RPLTSs with one action label ($Act =
\{a\}$).  Due to this, all d-trees are paths representing runs in the
Markov chain.  Thus the linear-time semantics of PCTL* carries over,
since all the ``trees'' degenerate to paths.  The GPL encoding of
PCTL* then relies on the following three basic steps:
\begin{enumerate}
\item Next state operator $\mathsf{X}$ is encoded in terms of the
  diamond modality $\diam{a}$ in GPL.
\item Until formulae $\psi_1 \mathsf{U} \psi_2$ are encoded by
  unrolling them as $\psi_2 \lor (\psi_1 \land \mathsf{X}(\psi_1
  \mathsf{U} \psi_2))$ and using a least fixed point GPL formula to
  represent the unrolling.
\item Formulae with negation of the form  $\neg \psi$ are encoded by
  finding negating the encoding of $\psi$.  
\end{enumerate}
Other operators including PCTL* path quantifiers have corresponding
operators in GPL, and are translated directly.

The above encoding has a significant limitation: although RPLTSs in
general can exhibit both nondeterministic and probabilistic choices,
the GPL-based encoding was only for model checking PCTL* only over
Markov chains, and not over MDPs.  This is because the nondeterminism
in MDPs has linear-time semantics, while all nondeterminism in GPL is
under the branching-time semantics.  Indeed, this nondeterminism is
entirely unused in the above encoding by limiting RPLTSs to one action
label ($Act = \{a\}$), which, in turn made every d-tree into a path.

As XPL semantics for fuzzy formulae is also defined over d-trees, the
PCTL* encoding of \cite[Sect.~3.2]{gpl} carries over to XPL
essentially unchanged.  For model checking MDPs, we represent them as
PLTSs, treating the internal nondeterminism among the actions in an
MDP as internal nondeterminism in a PLTS as well.  Note that d-trees
of a PLTS obtained from an MDP are still paths, since branching in the
d-trees only represents external nondeterminism.  Consequently, the
addition of internal nondeterminism in PLTSs, interpreted under the
linear-time semantics, is orthogonal to the problem of encoding PCTL*,
because the internal nondeterminism is resolved by the time we reach
d-trees.  Model checking of a PCTL* formula $\gamma$ over an MDP is
cast as XPL model checking of the corresponding PLTS, where the XPL
formula is generated by $E_{PCTL^*}(\gamma)$ defined in
Table~\ref{tab:pctl}.  In the definition, ``$\mathsf{neg}(\psi)$''
represents the negation of an XPL formula, also expressed in XPL.
Note that, as stated earlier, the translation of PCTL* formulae to XPL
formulae is virtually identical to the translation to
GPL~\cite[Sect~3.2]{gpl}.  The novelty is that we have identified the
XPL formulae resulting from our translation as separable, and hence
PCTL* properties can be successfully model checked over MDPs with our
XPL model checking algorithm.

\subsection{Encoding of RMDP Termination}\label{sec:rmdp}
We consider recursive MDPs (RMDPs)~\cite{rmdp} as a nondeterministic
extension of Recursive Markov Chains (RMCs)~\cite{RMC}.  We discuss a
more general model, called \emph{recursive simple stochastic games}
(RSSGs); formally, an RSSG $A$ is a tuple $(A_1, \dots, A_k)$, where
each \emph{component graph} $A_i$ is a septuple $(N_i, B_i, Y_i,
\kw{En}_i, \kw{Ex}_i, \kw{pl}_i, \delta_i)$:
\begin{itemize}
\item $N_i$ is a set of nodes, containing subsets $\kw{En}_i$ and
  $\kw{Ex}_i$ of entry and exit nodes, respectively.

\item $B_i$ is a set of boxes, with a mapping $Y_i \fcn{B_i} \{1,
  \dots, k\}$ assigning each box to a component.  Each box has a set
  of call and return ports, corresponding to the entry and exit nodes,
  respectively, in the corresponding components: $\kw{Call}_b =
  \{(b,en) \mid en \in \kw{En}_{Y_i(b)}\}$, $\kw{Return}_b = \{(b,ex)
  \mid ex \in \kw{Ex}_{Y_i(b)}\}$.  Additionally, we have:
  \begin{align*}
    \kw{Call}^i &= \bigcup_{b \in B_i} \kw{Call}_b, \\ \kw{Return}^i
    &= \bigcup_{b \in B_i} \kw{Return}_b, \\ Q_i &= N_i \cup
    \kw{Call}^i \cup \kw{Return}^i.
  \end{align*}

\item $\kw{pl}_i \fcn{Q_i} \{0, 1, 2\}$ is a mapping that specifies
  whether, at each state, the choice is probabilistic (i.e.,
  \emph{player} $0$), or nondeterministic (player $1$: maximizing,
  player $2$: minimizing).  As any $u \in \kw{Call}^i \cup \kw{Ex}_i$
  has no outgoing transitions, let $\kw{pl}_i(u) = 0$ for these
  states.

\item $\delta_i$ is the transition relation, with transitions of the
  form $(u, p_{uv}, v)$, when $\kw{pl}_i(u) = 0$ and $u$ is not an
  exit node or a call port, and $v$ may not be an entry node or a
  return port.  Additionally, $p_{uv} \in (0,1]$ and, for each $u$,
    $\sum\limits_{v':(u, \cdot, v') \in \delta_i} p_{uv'} = 1$.
    Meanwhile, the nondeterministic extension yields transitions of
    the form $(u, \bot, v)$ when $\kw{pl}_i(u) > 0$.
\end{itemize}

Recursive MDPs (RMDPs) only have a player $1$ or player $2$, depending
on whether they are maximizing or minimizing, respectively.
Termination probabilities can be computed for $1$-RSSGs, and are
always achieved, for both players, with a strategy limited to a class
called \emph{stackless} and memoryless (SM)~\cite{rmdp}.  The essence
of SM strategies is that in each nondeterministic choice, the
selection is fixed to a single state from its distribution, which
makes the resolution of the nondeterministic choices substantially
simpler than in the general case.  For multi-exit RSSGs, the
termination probability is \emph{determined}~\cite{rmdp}, although an
optimal strategy may not exist, and the problem of computing the
probability is undecidable, in general.  SM strategies are inadequate
even for $2$-exit RMDPs~\cite{rmdp}.  Figure~\ref{fig:rmdp} shows a
recursive MDP with two components, $A$ and $B$.  Any call to $A$
nondeterministically results in either a call to $B$ (via box $b_1$)
or a transition to $u$.

\subsubsection{Translating RMDPs to PLTSs}\label{sec:trans}
Given an RMDP $A$, we can define a PLTS $L$ that simulates $A$, with
$\kw{Act} = \{p, n, c, r_i, e_i\}$ and states of the PLTS
corresponding to nodes of the RMDP.  We retain the RMDPs transitions,
labeling them as $n$ for actions from a nondeterministic choice and
$p$ for probabilistic choice.  To this basic structure we add three
new kinds of edges:
\begin{itemize}
\item $e_i$ for the $i$th exit node of a component,

\item $c$ edges from a call port to the called component's entry node,
  and

\item $r_i$ edges from a call port to each return port in the box.
\end{itemize}
While $c$ edges denote control transfer due to a call, $r$ edges
summarize returns from the called procedure.  Figure~\ref{fig:rmdp}
shows the result of the translation for one component of the RMDP.
Formally, we define the PLTS $L$ as follows:
\begin{definition}[Translated RMDP]\label{def:trmc}
  The translated RMDP $A$ is a PLTS $L = (S, \delta, P, I)$:
  \begin{itemize}
  \item The set of states $S$ is the set of all the nodes, as well as
    the call and return ports of the boxes, i.e., $S = \bigcup_i Q_i$.
    Additionally, we associate a consistent index with each state
    corresponding to an exit node or a return port.

  \item The transition relation $\delta$ has all the transitions of
    the components, labeled by action $p$ for the probabilistic
    transitions and $n$ for the nondeterministic ones.  Thus, when
    $(u, p_{uv}, v) \in \delta_i$ for any $i$, then $(u, p, v) \in
    \delta$, and when $(u, \bot, v) \in \delta_i$, $(u, n, v) \in
    \delta$. Additionally, we have $((b,en), c, en) \in \delta$ and
    $((b,en), r_i, (b,ex_i)) \in \delta$ for every box $b$, and
    $(ex_i, e_i, ex_i) \in \delta$ for every exit node.  Note the
    indices used.

  \item The transition probability distribution $P$ is defined as
    $P(u, p, v, \cdot) = p_{uv}$ as given for the RMDP $A$, $P(u, n,
    v, c(v)) = 1$, where $c \fcn{S} \mathbb{N}$ is a one-to-one
    function (when $c \neq c(v)$ for any $v$ with $(u,n,v) \in
    \delta$, $P(u,n,v,c) = 1$ for an arbitrary $v$ with $(u,n,v) \in
    \delta$), and $P(\cdot) = 1$ if the action is not $p$ or $n$.

  \item We do not use the interpretation in the translation, i.e.,
    $I(s) = \emptyset$ for any state $s$, unless additional relevant
    information about the RMDP $A$ is available.
  \end{itemize}
\end{definition}
For RMCs, the translation yields a simulating RPLTS $L$ (no $n$
actions).

Intuitively, $L$ preserves all the non-recursive transition structure
of $A$ via the actions labeled by $p$ and $n$.  There are additional
$c$ actions to model call transitions.  Note that each call port will
have a single outgoing $c$ transition, while the entry nodes may have
multiple incoming $c$ transitions.  Meanwhile, we need a different
design to associate exit nodes with return ports, as an exit node may
be associated with multiple return ports.  Thus, we have indexed $e$
and $r$ actions and require a standard formula to model termination.
We note that the resulting structure is similar to the nested state
machines (NSM)~\cite{nsm}, with the $p/n$, $c$, $r_i$, and $e_i$ edges
corresponding to the \kw{loc} (local), \kw{call}, \emph{jump}, and
\kw{ret} edges, respectively, in the NSM model.
\begin{figure}
  \centering\includegraphics[scale=.6]{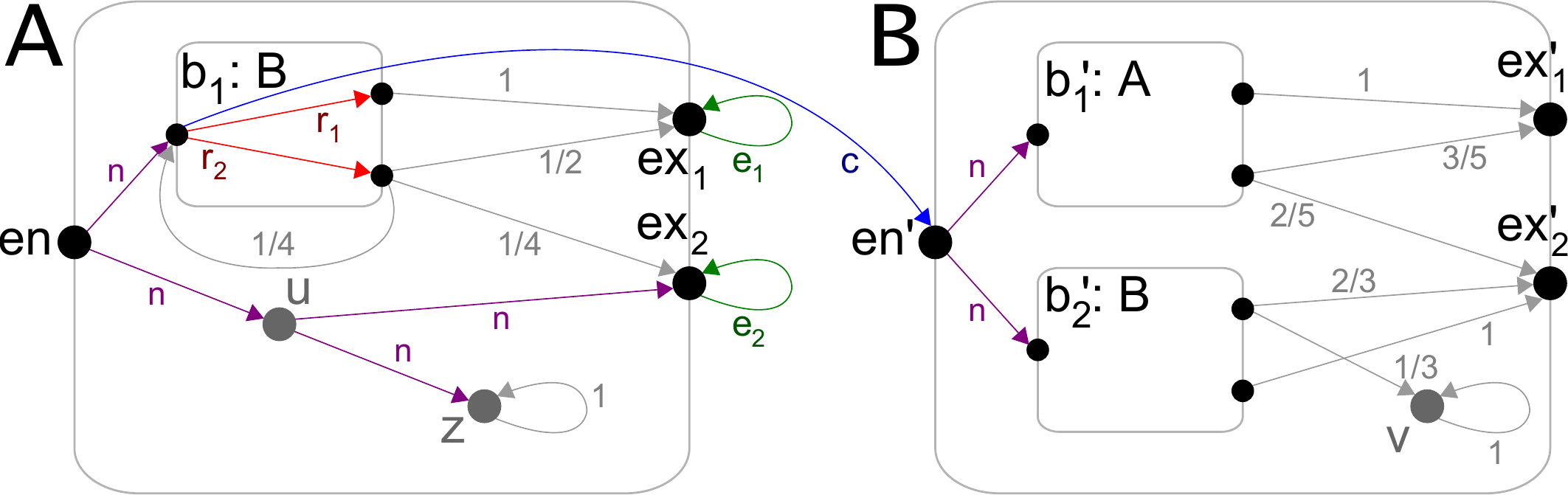}
  \caption{Example RMDP with Call, Return, and Exit edges added to
    $A$}\label{fig:rmdp}
\end{figure}

Termination of $1$-RMDPs can be encoded as the following separable
formula:
\begin{equation}\label{eqn:sermc}
  \psi_1 = \mu X.\diam{e_1} \mathsf{tt} \lor \diam{p} X \lor \diam{n}
  X \lor (\diam{c} X \land \diam{r_1} X)
\end{equation}

Termination of multi-exit RMDPs is undecidable, in
general~\cite{rmdp}.  We can still encode it in XPL, but the resulting
formula is not separable: the termination formula for a $2$-exit
RMDP~(\ref{eqn:mermc}) is entangled on the $c$ action.
\begin{align}\label{eqn:mermc}
    \psi_2^1 =_\mu \diam{e_1} \mathsf{tt} \lor \diam{p} \psi_2^1 \lor
    \diam{n} \psi_2^1 &\lor (\diam{c} \psi_2^1 \land \diam{r_1}
    \psi_2^1) \\ \nonumber &\lor (\diam{c} \psi_2^2 \land \diam{r_2}
    \psi_2^1) \\ \nonumber
    \psi_2^2 =_\mu \diam{e_2} \mathsf{tt} \lor \diam{p} \psi_2^2 \lor
    \diam{n} \psi_2^2 &\lor (\diam{c} \psi_2^1 \land \diam{r_1}
    \psi_2^2) \\ \nonumber &\lor (\diam{c} \psi_2^2 \land \diam{r_2}
    \psi_2^2)
\end{align}

We note that the conjunction $\diam{c}X \land \diam{r_1}X$
in~(\ref{eqn:sermc}) is independent.  Additionally, the disjunction
between the two conjuncts in~(\ref{eqn:mermc}) is mutually exclusive,
since $\diam{c} \psi_1^2$ and $\diam{c} \psi_2^2$ correspond to
eventually reaching distinct exits; however, it is correct to sum them
only for RMCs, as the nondeterministic choices in RMDPs preclude the
simple summation of mutually exclusive outcomes.

\subsection{PTTL and Branching Processes}\label{sec:pttl}
PCTL*~\cite{PCTL*} may be considered a linear-time logic, in the sense
that its fuzzy formulae are essentially full LTL.  Similarly,
PCTL~\cite{PCTL} is not the only plausible extension of CTL: instead
of replacing the $\mathsf{A}$ and $\mathsf{E}$ operators with the
$\mathsf{Pr}$ operators, we could have full CTL as fuzzy formulae, as
there is a natural interpretation of CTL over d-trees, and this logic,
over RPLTS, would be subsumed by GPL.

A similar logic, Probabilistic Tree Temporal Logic (PTTL), has been
independently introduced~\cite{bpmc}.  Branching Processes (BPs) are a
branching-time extension of Markov chains, and PTTL is a logic over
BPs.  The problem of BP extinction corresponds to termination of
$1$-exit RMCs~\cite{RMC}.  BPs have also been extended with
nondeterminism, yielding Branching MDPs (BMDPs), for which the
extinction and reachability problems have been analyzed~\cite{bmdp,
  rmdp}.

We write the syntax of PTTL~\cite[Definition~$18$]{bpmc}, where $A \in
Prop$, and we refer to $\phi$ and $\psi$ as state and fuzzy formulae,
as for XPL:
\begin{equation*}
  \begin{array}{l}
    \phi ::= A \mid \neg \phi \mid \phi \land \phi \mid \mathsf{Pr}_{>
      p} \psi \mid \mathsf{Pr}_{\geq p} \psi \\

    \psi ::= \mathsf{AX} \phi \mid \mathsf{EX} \phi \mid \mathsf{A}
         [\phi \mathsf{U} \phi] \mid \mathsf{E} [\phi \mathsf{U} \phi]
         \mid \mathsf{A} [\phi \mathsf{R} \phi] \mid \mathsf{E} [\phi
           \mathsf{R} \phi]
  \end{array}
\end{equation*}

In this section, we may view BPs as specialized RPLTSs, and RMDPs as
specialized PLTSs.  So, we give the semantics for PTTL over PLTSs
(without terminal states), assuming maximizing schedulers, by encoding
it in XPL, as $E_{PTTL}(\gamma)$, in Table~\ref{tab:pttl}.
\begin{table}
  \caption{Encoding of PTTL over BMDPs}\label{tab:pttl}
  \begin{equation*}
    E_{PTTL}(\gamma) = \left\{
    \begin{array}{ll}
      \hline

      \gamma, & \gamma \in Prop, \\

      \mathsf{neg}(E_{PTTL}(\gamma')), & \gamma = \neg \gamma', \\

      E_{PTTL}(\gamma_1) \land E_{PTTL}(\gamma_2), & \gamma = \gamma_1
      \land \gamma_2, \\

      \mathsf{Pr}_{> p} E_{PTTL}(\psi), & \gamma =
      \mathsf{Pr}_{> p} \psi, \\

      \mathsf{Pr}_{\geq p} E_{PTTL}(\psi), & \gamma =
      \mathsf{Pr}_{\geq p} \psi, \\

      [-] E_{PTTL}(\phi), & \gamma = \mathsf{AX} \phi, \\

      \diam{-} E_{PTTL}(\phi), & \gamma = \mathsf{EX} \phi, \\

      \mu X.E_{PTTL}(\phi_2) \lor (E_{PTTL}(\phi_1) \land [-] X), &
      \gamma = \mathsf{A}[\phi_1 \mathsf{U} \phi_2], \\

      \mu X.E_{PTTL}(\phi_2) \lor (E_{PTTL}(\phi_1) \land \diam{-} X), &
      \gamma = \mathsf{E}[\phi_1 \mathsf{U} \phi_2], \\

      \nu X.E_{PTTL}(\phi_2) \land (E_{PTTL}(\phi_1) \lor [-] X), &
      \gamma = \mathsf{A}[\phi_1 \kw{R} \phi_2]. \\

      \nu X.E_{PTTL}(\phi_2) \land (E_{PTTL}(\phi_1) \lor \diam{-} X), &
      \gamma = \mathsf{E}[\phi_1 \kw{R} \phi_2]. \\

      \hline
    \end{array}
    \right.
  \end{equation*}
\end{table}
This translation also concretely demonstrates how the branching-time
nature of GPL and XPL has not been recognized: either unnoticed (``the
existing model-checking algorithms do not work for branching
processes''~\cite{bpmc}) or misunderstood (``it cannot express $\dots$
the CTL formula $\mathsf{EG} p$''~\cite{castro} --- but, through PTTL,
it can).

%---------------------------------------------------------------------
\section{Discussion and Future Work}\label{sec:concl}
Previous attempts to extend GPL included allowing systems with
internal nondeterminism while still resolving the probabilistic
choices first~\cite{bran}, and EGPL, which had similar syntax and
semantics to XPL, but limited the model checking to non-recursive
formulae~\cite{soni}.

Following GPL, XPL treats conjunction in a traditional manner,
retaining the properties that $\psi \land \neg \psi = \kw{ff}$, and
$\psi \land \psi = \psi$ for any formula $\psi$.  However, the
probability value of $\psi_1\land \psi_2$ cannot be computed based on
the probability values of the conjuncts $\psi_1$ and $\psi_2$.  This
makes model checking in XPL more complex, but also contributes to its
expressiveness.

Another probabilistic extension of $\mu$-calculus is pL$\mu$.  In
contrast to XPL, the most expressive version of pL$\mu$, denoted
pL$\mu^{\odot}_{\oplus}$~\cite{indp, mio}, defines \emph{three}
conjunction operators and their duals such that their probability
values can be computed from the probabilities of the conjuncts.  The
logic pL$\mu^{\odot}$ is able to support branching time and an
intuitive game semantics~\cite{indp}.  Along the same lines as our XPL
encoding, we can encode termination of $1$-exit RMDPs as model
checking in pL$\mu^{\odot}$, and RMC termination in
pL$\mu^{\odot}_{\oplus}$.  However, attempting to encode multi-exit
RMDP termination in pL$\mu^{\odot}_{\oplus}$ similarly to multi-exit
RMC termination would lead to an incorrect, rather than undecidable,
encoding.  Determining the relationship between XPL and pL$\mu^\odot$
in branching time is an important problem.  Other recent probabilistic
extensions of $\mu$-calculus include the Lukasiewicz
$\mu$-calculus~\cite{lukas} and $\mu^p$-calculus~\cite{castro}, which
can encode PCTL* over MDPs, and P$\mu$TL~\cite{pmutl}, but all these
limit nondeterminism to the linear-time semantics.  Quantitative
$\mu$-calculi, such as qM$\mu$~\cite{McIver} and Q$\mu$~\cite{qmmu},
are more similar to pL$\mu$, so we do not offer an independent
comparison to XPL.

Although closely related, algorithms to check properties of RMCs (and
pPDSs~\cite{EKM:LICS04}) were developed independently~\cite{RMC}.
These were related to algorithms for computing properties of systems
such as branching process (BP) extinction and the language probability
of Stochastic Context Free Grammars.  The relationship between GPL and
these systems was mentioned briefly in~\cite{pip}, but has remained
largely unexplored.

There has been significant interest in the study of expressive systems
with nondeterministic choices, such as RMDPs and Branching MDP
(BMDPs)~\cite{rmdp}.  At the same time, the understanding of the
polynomial systems has expanded.  In~\cite{pps}, the class of
Probabilistic Polynomial Systems (PPS) is introduced, which
characterizes when efficient solutions to polynomial equation systems
are possible even in the worst case~\cite{newton}.  While~\cite{RMC}
did not distinguish the systems arising from $1$-exit RMCs from those
from multi-exit RMCs, the PPS class is limited to $1$-exit RMCs.  It
was also extended for RMDP termination and, later, BMDP reachability,
both having polynomial-time complexity for min/maxPPSs~\cite{pps,
  bmdp}.

Systems producing equations in PPS form show an interesting
characteristic: that the properties are expressible as purely
conjunctive or purely disjunctive formulae.  Recall that such formulae
are trivially separable.  Polynomial systems equivalent to those
arising from separable GPL have recently been considered in a more
general setting in the context of \emph{game automata}~\cite{tree15},
followed by an undecidability result for more general properties on
the automata~\cite{undec16}.  Characterizing equation systems that
arise from separable formulae and investigating their efficient
solution is an interesting open problem.  Finally, this paper
addressed the decidability of model checking; determining the
complexity of model checking is a topic of future research.

%%
%% Bibliography
%%

\end{document}